\documentclass[11pt]{article}
\usepackage[T1]{fontenc}
\usepackage{graphicx}

 \usepackage[margin=25mm]{geometry}

\usepackage[numbers,sort]{natbib}
\usepackage{algorithm}
\usepackage[noend]{algpseudocode}
\usepackage{enumerate}
\usepackage{multirow}

 \usepackage{amsthm}
\usepackage{amsmath}
\usepackage{xspace}
\usepackage{hyperref}
\usepackage[capitalise]{cleveref}
\usepackage{amssymb}

 \newtheorem{theorem}{Theorem}
 \newtheorem{lemma}{Lemma}
 \newtheorem{claim}{Claim}

\DeclareFontShape{OT1}{cmr}{bx}{sc}{<-> cmbcsc10}{}
\newcommand{\cgst}{\textsc{congest}\xspace}
\newcommand{\local}{\textsc{local}\xspace}
\newcommand{\clq}{\textsc{clique}\xspace}
\newcommand{\kc}{$k$-center\xspace}
\newcommand{\id}{id\xspace}
\newcommand{\len}{\ell}
\newcommand{\bin}{\textnormal{bin}}
\newcommand{\dist}{\textnormal{d}}
\newcommand{\optOne}{\textnormal{OPT}_1}
\newcommand{\opt}{\textnormal{OPT}_k}
\newcommand{\alg}{\textnormal{ALG}}
\newcommand{\ecc}{\textnormal{ecc}}

\DeclareMathOperator{\poly}{poly}

 \author{
 Leyla Biabani, Eindhoven University of Technology\\
 Ami Paz, LISN --- CNRS \& Paris-Saclay University}

 \date{}

\begin{document}

\title{$k$-Center Clustering in Distributed Models} 

\maketitle

\begin{abstract}
The $k$-center problem is a central optimization problem with numerous applications for machine learning, data mining, and communication networks.
Despite extensive study in various scenarios, it surprisingly has not been thoroughly explored in the traditional distributed setting, where the communication graph of a network also defines the distance metric.

We initiate the study of the  $k$-center problem in a setting where the underlying metric is the graph's shortest path metric in three canonical distributed settings: the \local, \cgst, and \clq models. 
Our results encompass constant-factor approximation algorithms and lower bounds in these models, as well as hardness results for the bi-criteria approximation setting.

\medskip\noindent\textbf{keywords:} $k$-Center clustering;
Distributed graph algorithms; 
Shortest path metric

\end{abstract}

\section{Introduction}

\subsection{Distributed $k$-Center}
The $k$-center problem is a key optimization problem, 
seeks to locate a small set of points in space (``centers'') that minimize the maximal distance from them to any other point.
The problem was extensively studied in the centralized setting, where the points are taken from a metric space, with or without guarantees on the metric.
Over the years, there was also some work on the $k$-center problem in parallel and distributed models, looking to solve the problem faster using multiple computational units.
As in the centralized setting, the points and metric under consideration in these works are taken from some metric space, and it is usually assumed that distances are given to the computational units by an oracle.

A natural setting for the \kc problem is that of a network.
We consider a graph representing a communication network in a natural way, 
and the goal is to make $k$ of the nodes into centers while minimizing the maximal distance between these nodes and all others.
Here, the metric space is not arbitrary, but is the \emph{graph metric} on the graph itself.
This is significant, for example, when determining server placement within a network to minimize the maximum delay, as seen in scenarios like content distribution over the Internet.

Two related, yet different problems are the metric \emph{facility location} problem, 
where the goal is to minimize the average distance from the servers (or equivalently, the sum of distances),
and the \emph{online $k$-server} problem, where the points and servers can move in space.

In this work, we initiate the study of \kc in the distributed setting, where an {undirected} graph represents both the communication network and the problem's metric.
We address the problem in the popular distributed models of \local, \cgst, and \clq, and derive upper and lower bounds.
While the problem was studied earlier in some of these distributed models, it is important to note that all the prior work considered points in an arbitrary metric space, independent of the communication graph.
As far as we know, our work is the first to consider the natural setting where the input graph represents both the network used for communication and the metric on which \kc is solved.

\subsection{Our Results and Techniques}
We present upper and lower bounds for the $k$-center problem in different distributed models. 
The metric considered is the shortest-path metric, and the graphs are unweighted unless otherwise specified.

\subsubsection*{The \local Model}
A simple and natural model for communication networks is the \local model~\cite{Linial92},
where the network's nodes communicate by synchronously exchanging messages of unbounded size.
The local computation time is neglected, and the complexity measure in this case is the number of communication rounds.
In this model, any problem is solvable in $D$ rounds, where $D$ is the diameter of the graph, i.e., the largest distance between two nodes in it.
Hence, our upper bounds should be read as the minimum of the given value and $D$, and the lower bounds only apply for graphs with a diameter larger than the lower bound value.

We start \cref{sec:local} by giving a relatively simple algorithm that finds a $(2k+\epsilon)$-approximate solution to the \kc problem in $O(k/\epsilon)$ rounds of the \local model.
This result relies on the assumption that the nodes have unique IDs in ${1,\ldots,n}$ --- if we assumed the range to be larger, even ensuring a specific number of centers would require global communication.
We then show that reducing the approximation ratio to be below $k-1$ requires $\Omega(n)$ time.
We stress that $D\leq n$,
so in $\Omega(n)$ time the nodes can aggregate the full structure of the network and find an exact solution.

Put differently, our results present two extremes.
On the one hand, if a large approximation ratio of at least $2k+\epsilon$ is allowed, the problem is simple---for~$\epsilon,k$ constants, it is solvable in constant time.
On the other hand, 
if a lower approximation ratio is required, e.g., a constant that does not depend on $k$, then the running time is so large that the nodes may as well compute an optimal solution.

\subsubsection*{The \cgst Model}
The \cgst model is a restrictive variant of the \local model, 
where the messages are limited to $O(\log n)$ bits.
Specifically, problems are no longer trivially solvable in $O(D)$ time,
and getting a constant approximation is solvable in non-trivial time, but is challenging.

In \cref{sec:congest} we present a $2$-approximation \cgst algorithm for the \kc problem running in $O(kD)$ rounds.
It constructs different BFS trees, sometimes from multiple sources simultaneously, and simulates a centralized greedy approximation algorithm for the \kc problem.
On the other hand, in \cref{sec:cgst lb}
we prove that improving the approximation ratio to below $4/3$, and hence also finding an exact solution, requires a much longer time:
any algorithm for this problem, even randomized, must take $\tilde\Omega(n/k)$ rounds,
and this is true even for graphs of diameter as small as 12.

This lower-bound proof is rather involved: it uses a reduction to communication complexity with a twist.
Proving \cgst lower bounds by reducing them to communication complexity is common in the literature~\cite{AbboudCKP21,SarmaHKKNPPW12,PelegR00}, where usually the solution for the \cgst problem directly implies an answer communication complexity problem. 
In our reduction, a new post-processing phase is added, where the players must do extra computation and also communicate more
after getting the solution to the (approximate) \kc problem and before finding an answer to the communication problem.

\subsubsection*{The \clq Model}
Finally, in \cref{sec:clique} we consider the more recent congested clique model, denoted \clq.
It resembles the \cgst, but with an all-to-all communication---the communication graph is a clique, and the input graph is a subgraph of it on the same set of nodes.
One might think that the $O(kD)$-round \cgst  algorithm will directly translate to an $O(k)$-round \clq algorithm, as the communication network has a diameter $1$.
However, this it not the case: the \cgst algorithm builds BFS trees in $O(D)$ time, and this step cannot be translated to an $O(1)$-round subroutine in the \clq model.

At a high level, we show that finding a $2$-approximate solution for the 
\kc problem in this model can be done deterministically in the same time as computing all-pairs-shortest paths,
which is $O(n^{1/3})$ for weighted graphs, and $O(n^{0.158})$ for unweighted graphs~\cite{Censor-HillelKK19}.
Interestingly, if~$\omega$, the exponent of the (centralized) matrix multiplication problem, will be discovered to be $2$, the time for distributed matrix multiplication algorithms $2$-approximation of \kc will also be improved, to $O(n^{\epsilon'})$ for any $\epsilon'>0$, or even to be poly-logarithmic.
Previous work implies that a $(2+\epsilon)$-approximate solution for \kc can be found in only
$O(\poly\log n)$ rounds, even in weighted graphs~\cite{BandyapadhyayIP22}.
By allowing higher approximation ratios we can get even faster algorithms, 
such as a $(4+\epsilon)$-approximation in $O(\poly\log\log n+k)$ rounds,
and up to an $O(\log n)$-approximation in $O(k)$~rounds.

En route, we prove a new result regarding approximate \kc.
A simple greedy algorithm~\cite{DBLP:journals/tcs/Gonzalez85}
finds a $2$-approximation of \kc, assuming the distance between every two nodes is known.
We extend this claim in \cref{lem:greedy:2+eps}, to show that if only a one-way, multiplicative $\alpha$-approximation of the distances is known,
then a similar greedy algorithm gives a $(2\alpha)$-approximate solution to the \kc problem.
Hence, by studying distributed \kc, we also provide new insights to the centralized case, which could be of independent interest.

Unfortunately, current techniques cannot establish lower bounds in the \clq model. 
Any non-trivial lower bound in this model will imply circuit complexity lower bounds, solving a long-standing and notoriously hard open problem~\cite{DruckerKO13}.

\section{Preliminaries}
\subsection{The $k$-Center Problem}
Consider a graph $G=(V, E)$ with nodes $V$ and edges $E$ where $|V|=n$. 
The edges in set $E$ may have weights or be unweighted, and we mainly focus on the former case, and state it explicitly when this is not the case.
The length of a path is determined by the number of edges in the path when $E$ is unweighted, and by the total weight of the edges in the path when $E$ is weighted.
The diameter of the graph $G$, denoted  $D$, is the maximum distance between any two nodes in $V$, i.e., $D=\max_{u,v\in V} \dist(u,v)$. 
For any node $u \in V$, the eccentricity is $\ecc(u) = \max_{v \in V} \dist(u, v)$, so we also have $D=\max_{u \in V}\ecc(u)$.

Let $k \in \mathbb{N}$ be a given parameter. In the $k$-center problem, we aim to find a set $S^* \subseteq V$ with a size of at most $k$ that minimizes 
the maximum distance of any node of $V$ to its
nearest center in~$S$. 
More formally, we seek a set $S^*$ with $|S^*| \leq k$ that minimizes the value $\max_{v \in V} {\min_{s \in S^*} \dist(v, s)}$. This value is denoted as $\opt(G)$ or simply $\opt$ when the context is clear.
For $\alpha \geq 1$, an algorithm is considered $\alpha$-approximation if it can compute a solution $S$ such that the distance of any node $v \in V$ to its nearest node in $S$ is at most~$\alpha \cdot \opt$.

\subsection{Computational Models}
We consider three common computational models for studying distributed graph algorithms, namely the \local, \cgst and \clq models.
We model a communication network using its graph, with nodes representing computational units and edges representing communication links.
We use $n$ for the number of nodes (computational units) and assume they have unique ids in $\{1,\ldots,n\}$ 
(specifically, there is always a node with \id 1).
The computation proceeds in synchronous rounds, where in each round each node sends messages to its neighbors, receives messages form them, and updates its local state accordingly.
The input is local, in the sense that each node initially knows only its id, list of neighbors, and if there are inputs such as edge weights, then also the weights of its incident edges.
The outputs are similarly local, e.g., at the end of the algorithm's execution each node should know if it is a center or not.

In the \local model~\cite{Peleg00}, the message sizes are unbounded, and an $r$-round algorithm is equivalent to having each node deciding by its distance-$r$ neighborhood~\cite{Linial92}.
The \cgst model~\cite{PelegR00,SarmaHKKNPPW12} is similarly defined, but each node is limited to sending $O(\log n)$-bit messages to each of its neighbors in each round, where messages to different neighbors might be different from one another.
This model allows each node to send, e.g., its id or the ids of some of its neighbors in a single round, but not its full list of neighbors.
A common primitive in this model is that of construction a BFS tree from a node.
This is sometimes extended to multiple BFS trees, or trees of bounded depth. 
See, e.g.,~\cite{LenzenPP19,HolzerW12}.

Finally, we model a network with a congested all-to-all overly network by the \clq model~\cite{LotkerPPP03,DruckerKO13}.
In this model, the input is a network as before, but the communication is less limited: in each round, each node can send $O(\log n)$-bit messages to each other node in the graph, and not only to its neighbors.
This allows the nodes, e.g., to re-distribute the inputs in constant time~\cite{DolevLP12,Lenzen13}, and compute all-pairs-shortest-paths in sub-linear time~\cite{Censor-HillelKK19}.

\subsection{Communication Complexity}
We prove lower bounds for the $k$-center problem in the \cgst model using a reduction to \emph{communication complexity},
a well-studied topic in theoretical computer science~\cite{KushilevitzN:book96,RaoY20}.

In the two-party \emph{set disjointness} (henceforth: disjointness) communication complexity problem,
two players referred to as Alice and Bob get two $\ell$-bit strings $x$ (Alice's string) and $y$ (Bob's string), and communicate by exchanging messages on a reliable asynchronous channel.
Their goal is to decide if the sets represented by the indicator vectors $x$ and $y$ are disjoint or not, 
i.e., if there is an index $i$ such that $x[i]=y[i]=1$, in which case we say they are not disjoint and the players must output $0$, 
or otherwise, the sets are disjoint and they should output $1$.

Alice and Bob follow some protocol indicating who should send messages at each step and what message to send.
The communication complexity of a protocol (as a function of $\ell$ is the maximal number of bits they exchange when executing it, 
and the deterministic communication complexity of a problem the the minimal communication complexity of a protocol solving the problem.
The randomized communication complexity is similarly defined, but the players may also use random bits when executing the protocol, and the success probability is at least $2/3$, when taken on the choice of random bits.

For the disjointness problem, there is an $\Omega(\ell)$ lower bound,
which holds for deterministic and randomized algorithms alike~\cite[Example 3.22]{KushilevitzN:book96}.

\subsection{Related Work}
To our surprise, the \kc problem was not studied in traditional distributed computing models.
We survey below works in related computational models, and works on the related problem of metric facility location in distributed settings.

\subsubsection{$k$-Center in Related Computational Models}
\citet*{BandyapadhyayIP22} studied the \kc problem, along with the related uncapacitated facility location and $k$-median problems.
They consider the $k$-machine model (the parameter $k$ here need not be the same as in \kc) which is closely related to the \clq model.
As mentioned, their work implies a randomized $(2+\epsilon)$-approximation $O(\poly\log n)$-time \kc algorithm in the \clq model.
\citet*{ChiplunkarKR20} studied approximate \kc in \emph{streaming} models, and also in a parallel model where multiple processors perform local computations and then send the results to a central processor.
Surprisingly, these seem to be the only works on the \kc problem in distributed settings.
\citet*{CrucianiFGNS24}
recently studied the \kc problem in a centralized \emph{dynamic graphs} setting.

\subsubsection{Metric Facility Location}
The metric facility location problem has attracted more attention in distributed settings than the \kc problem.
Works on this problem consider models that share similarities with the ones we consider here, although typically the models are not precisely identical.

\paragraph{The \cgst Model}
In an unpublished manuscript, \citet*{BriestDKKP11} studied 
metric facility location in the \cgst model.
They focused on a bipartite graph, with one side representing facilities and the other representing clients.
Some of their results appeared in the thesis of 
\citet*{Pietrzyk13}.
This work improved upon previous work that also consider metric facility location in bipartite setting in the \cgst model~\cite{PanditP09,MoscibrodaW05}.

\paragraph{The \clq Model}
\citet*{GehweilerLS14} studied metric facility location in a model resembling the \clq.
They present a $3$-round randomized algorithm that gives a constant approximation factor,
based on the method of Mettu and Plaxton~\cite{MettuP03}.

\paragraph{Distributed Large-Scale Computational Models}
\citet*{InamdarPP18} studied metric facility location in the \clq model, MPC and $k$-machine, and gave an $O(1)$-approximation algorithm using a Mettu-Plaxton-style algorithm. 
When considering the \clq model with what they call ``implicit metric'', their model coincides with ours.
They also consider another input regime (``explicit metric''),
but not in the \local and \cgst models, and not for the $k$-center problem.
Their work improves upon earlier works on 
metric facility location in the \clq model~\cite{HegemanP15,BernsHP12}.

\section{The $k$-Center Problem in the \local Model}
\label{sec:local}
In this section, we consider the $k$-center problem in the \local model.
We start with a fast and  simple $((2+\epsilon)k)$-approximation algorithm, \cref{alg:local}.

\begin{lemma}
For any $\epsilon>0$ there is a deterministic $O(k/\epsilon)$-round 
algorithm in the \local model that gives a $((2+\epsilon)k)$-approximate solution for the $k$-center problem.
\end{lemma}

 \begin{algorithm}[tb]
\caption{A $((2+\epsilon)k)$-approximation in the \local model}\label{alg:local}
\begin{algorithmic}[1]
    \State $t\gets 2+\frac{4}{\epsilon}$
    \State Perform a BFS from $v_1$ for $tk$ rounds
    \State Report back on the tree if terminated on all branches under you
    \If{All branches terminated}
    \State Aggregate all the graph to $v_1$
    \State Locally compute an optimal solution
    \State Disseminate the solution on the BFS tree
    \Else{} Only $v_1$ marks itself as a center
    \EndIf
\end{algorithmic}
\end{algorithm}

\begin{proof}
Let $t=2+\frac{4}{\epsilon}$.
The algorithm starts by having the node with \id $1$ (or any other arbitrary node) initiating the construction of a BFS tree of depth $tk$. 
After $tk$ rounds, 
the nodes report back up the tree if the tree construction algorithm terminated in all its branches (i.e., reached all the nodes) or not.
If it terminated, we have $D\leq tk$,
and in another $O(tk)$ rounds node 1 can aggregate all the graph structure and find an optimal solution.
Otherwise, node 1 becomes the only center.

We claim that this simple algorithm gives a $(2+\epsilon)k$ approximate solution;
clearly, we only need to prove it for the case $D>tk$.
The eccentricity of node 1 is at most $D$, which gives an upper bound on $\alg$, the quality of the solution given by the algorithm.

Recall that $\opt$ is the largest distance from a node to its center in an optimal solution.
Let $\pi$ be the shortest path between two nodes of the largest distance, i.e., two nodes of distance $D$.
In an optimal solution, there are at least $D+1-k$ non-center nodes in $\pi$. Therefore, there is a center that covers at least $(D+1-k)/k$ nodes of $\pi$.
Hence, $\opt$ is at least $(D+1-k)/(2k)=D/2k+1/2k-1/2\geq D/2k -1$.

Note that $D>tk$ implies $D/(tk)>1$.
The approximation ratio is thus 
\[\frac{\alg}{\opt}
\leq\frac{D}{D/2k-1}
\leq\frac{D}{D/2k-D/(tk)}
=\frac{2tk}{t-2}
=(2+\epsilon)k
\]
by the choice of $t$, as required.
\end{proof}

This algorithm utilizes the fact that local computation is not limited in the \local model, so $v_1$ can find an optimal solution.
However, it is not strictly necessary to use all this computational power: in case $v_1$ has to locally solve \kc, it can instead compute a $2$-approximation in a greedy manner.
The approximation ratio is still as required, and the local computation now takes polynomial time.

Next, we move to the main result of the section: a lower bound for the $k$-center problem in the \local model. 
Our lower bound also works for bi-criteria algorithms, which return at most $\beta k$~centers instead of at most $k$ centers, for some $\beta\geq1$;
if the algorithm returns at most $k$ centers, we can simply set $\beta=1$. 
Our lower bound states that any $t$-round algorithm cannot achieve an approximation ratio better than $k-\frac{k^2+k(\beta k-1)(2t+1)}{n+k}$. 
If $k^2+k(\beta k-1)(2t+1) < n+k$ holds, then our lower bound states that it is not possible to get an approximation ratio better than $k-1$. 
The typical case is the non-bi-criteria one, with $k$ being a constant, 
and there our lower bound translates to stating that getting an approximation ratio better than $k-1$ requires
linear time.

\begin{theorem}
\label{thm:lower:bound:local}
    Let $t, k \in \mathbb{N}$, $\beta\geq 1$.
    Any $t$-round bi-criteria deterministic \local algorithm that solves the $k$-center problem and reports at most $\beta k$ centers as the solution, cannot have an approximation ratio better than $ k-\frac{k^2+k(\beta k-1)(2t+1)}{n+k}$, where $n > 2\beta kt$ is the size of the graph and $\beta \geq 1$.  
\end{theorem}

At a high-level, the proof considers a communication graph which is a cycle, an algorithm in the \local model that is faster than the lower bound, and the centers chosen by it.
In $t$ rounds, each node (and specifically, each chosen center) can gather information only from a segment of $2t+1$ nodes.
We thus create a new cycle by concatenating all the $(2t+1)$-node segments around the chosen centers.
By the choice of the parameters, this leaves some ``leftover nodes'' that are not in any segment, which are concatenated after all the segments.
When executing the same algorithm on the new cycle, the same nodes as before will become centers, as they gather exactly the same information. 
The ``leftover nodes'' are now far from all the chosen centers, rendering a bad approximation ratio, as claimed.

\begin{proof}
     Let $\mathcal{A}$ be an algorithm that finds an (approximate, bi-criteria) solution for the $k$-center problem.
     Let $k' \leq \beta k$ be the maximum number of centers that $\mathcal{A}$ reports for any cycle of length $n$,  and let $C$ be such a cycle.
     For simplicity, we label the nodes of $C$ with numbers $1,\ldots,n$ in a clockwise order. 
     Assume that $c_1<\ldots<c_{k'}$ are the centers returned by $\mathcal{A}$ in clockwise order. 
     Since $n > 2kt$, there exists at least two consecutive centers such that the distance between them is more than $2t$.
    We introduce a new cycle $C'$ by re-arranging the nodes and show that the solution that $\mathcal{A}$ finds for $C'$
    has an approximation ratio of at least $ k-\frac{k^2+k(\beta k-1)(2t+1)}{n+k}$.
    
    To build the cycle $C'$, we first define segments $S_1,\ldots,S_{k'}$ of $C$ as follows.
    Each segment $S_i$ is of the form $S_i=[b_i,e_i]$, which refers to the nodes $b_i$ to $e_i$ in the clockwise order. 
    Roughly speaking, each $S_i$ is a segment of $2t+1$ nodes centered around $c_i$, but since such segments might overlap, the exact definition is a bit more subtle.
    
    We define $e_i:=\min(c_i+t, c_{i+1}-1)$ for any $1\leq i<k'$ and $e_{k'}:=\min(c_{k'}+t, n+c_1-1)$.
     We also define $b_1:=\max(c_1-t, n-c_{k'}+1)$ and $b_i := \max(c_i-t,e_{i-1}+1)$ for any $1<i\leq k'$.
    Note that $e_k$ may be larger than $n$, in which case we consider $S_{k'}$ as the segment starting at $b_i$ and ending at $e_{k'}-n$ in the clockwise order, i.e. $S_{k'}=[b_{k'},n]\cup [1,e_{k'}-n]$.
    Similarly, $b_1$ might be smaller than $1$, in which case we consider nodes $b_1+n$ to $e_1$ in the clockwise order, i.e. $S_1=[b_1+n,n] \cup [1,e_1]$. 

    Since we assumed $n>2\beta kt$, then there is an $i^*$ such that the distance between $c_{i^*}$ and $c_{i^*+1}$ is more than $2t$ (we consider $c_1$ as $c_{i^*+1}$ if $i^*=k'$).
    To build $C'$, we concatenate the segments $S_{i^*+1}, S_{i^*+2},\ldots,S_k$ followed by $S_{1},S_{2},\ldots,S_{i^*}$, and finally all the remaining nodes of $C$, in a clockwise order (see Figure~\ref{fig:cycle}).

    \begin{figure}
    \begin{center}
    \includegraphics[width=\linewidth]{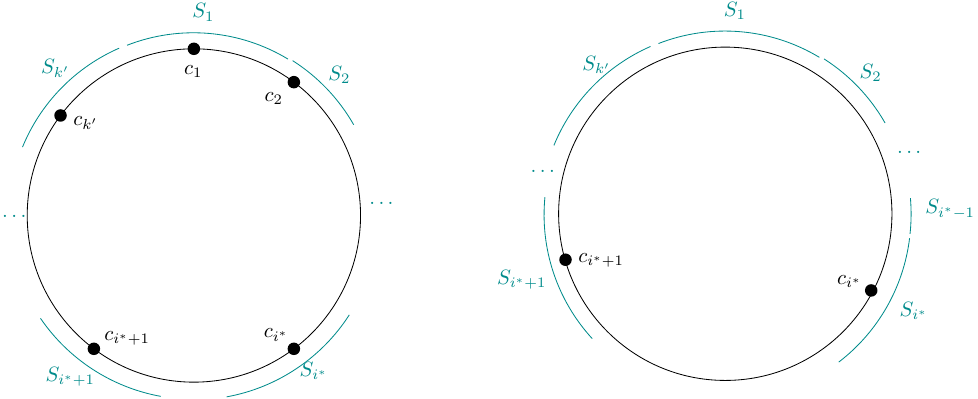}
    \end{center}
    \caption{Illustration for the proof of Theorem~\ref{thm:lower:bound:local}. Left: cycle $C$. Right: cycle $C'$}
    \label{fig:cycle}
    \end{figure}
     
     Now, we claim that there is an execution of  $\mathcal{A}$ on $C'$ that reports $c_1,\ldots,c_{k'}$ as the centers. 
     Since the distance between $c_{i^*}$ and $c_{i^*+1}$ is more than $2t$, then the neighbourhood of length $t$ for the nodes $c_{i^*}$ and $c_{i^*+1}$ is disjoint. 
     Along with the definition of the segments, 
     we can conclude that the distance-$t$ neighborhood for any node $c_i$ is the same in both $C$ and $C'$. 
     Hence, any node $c_i$ receives the same information for both graphs $C$ and $C'$, and therefore, 
     each node $c_i$ makes the same decisions.
     Since $c_1,\ldots,c_k$ chose to be the centers in $C$, they also chose to be the centers in $C'$. 
     Moreover, we assumed that $k'$ is the maximum number of centers that $\mathcal{A}$ reports for a cycle of length $n$, and hence,
    no other node chooses to be a center.

     We next discuss the approximation ratio. 
     According to the definition of $S_i$, each $S_i$ is of size at most $2t+1$. Therefore, there are at most $(k'-2)(2t+1)$ nodes in $S_{i^*+2},S_{i^*+3}\ldots,S_{k}, S_{1}, S_{2}, \ldots,S_{i^*-1}$. 
     Then, in the clockwise path from $c_{i^*+1}$ to $c_{i^*}$ in $C'$ there are at most $(t+1)+(k'-2)(2t+1)+(t+1)=(k'-2)(2t+1)+2t+2$ nodes. 
     This means that the distance from $c_{i^*+1}$ to $c_{i^*}$ in $C'$ in the anti-clockwise order is at least $n-((k'-2)(2t+1)+2t+2)+1=n-(k'-1)(2t+1)$, and 
     note that there is no center in this path. 
     Hence, there is a node in the anti-clockwise path from $c_{i^*+1}$ to $c_{i^*}$ in $C'$ such that its distance to the nearest center is at least $(n-(k'-1)(2t+1))/2$.
     
     On the other hand, the optimal solution has
     $\opt\leq\lceil(n-k)/(2k)\rceil$. To see this, we can choose the $k$ centers such that the distance between any two consecutive centers in the clockwise order is  
     $\lceil n/k\rceil$ or $\lfloor n/k\rfloor$.
     Thus, $\opt \leq \lceil(n-k)/(2k)\rceil$ holds 
     and the approximation ratio of Algorithm $\mathcal{A}$ is at least
     \begin{align*}
     &\frac{(n-(k'-1)(2t+1))/2}{\lceil(n-k)/(2k)\rceil} 
     \geq \frac{(n-(k'-1)(2t+1))/2}{(n-k)/(2k)+1}  
     = k\frac{n-(k'-1)(2t+1)}{n-k+2k} \\
     &\quad \quad  \quad  
     = k-\frac{k^2+k(k'-1)(2t+1))}{n+k}
     \geq  k-\frac{k^2+k(\beta k-1)(2t+1)}{n+k} \enspace,
    \end{align*}
     which finishes the proof.
\end{proof}

\section{A $2$-Approximation in the \cgst Model}
\label{sec:congest}
In this section, we show how to achieve a $2$-approximate $k$-center clustering in the \cgst model in $O(kD)$ rounds, where $D$ is the diameter of the underlying graph. Algorithm~\ref{alg:congest} presents an overview of our technique. In the following, we explain each part of this algorithm in detail.

\begin{algorithm}
\caption{A $2$-approximation in the \cgst model}\label{alg:congest}
\begin{algorithmic}[1]
    \State Find the node $v_{\min}$ with minimum \id
    \State $S\gets \{v_{\min} \}$
    \For{$k-1$ times}
        \State Perform a BFS from all the nodes in $S$
        \State Let $v^*$ be the furthest node from $S$, breaking ties by \id
        \State $S \gets S \cup \{v^*\}$
    \EndFor
    \State Each node in $S$ marks itself as a center
\end{algorithmic}
\end{algorithm}
As opposed to the \local algorithm, here we do not need to use the assumption that the nodes ids are in $\{1,\ldots,n\}$, and assuming they are taken from $\{1,\ldots,\poly n\}$ suffices.
The first step of Algorithm~\ref{alg:congest} is finding the node $v_{\min}$ with the minimum \id, which can be done in $O(D)$ rounds (and can be skipped if the ids are in $\{1,\ldots,n\}$). 
To do this, we start a BFS from all nodes.
In each round, each node may receive messages of BFS trees from multiple sources.
If this happens, such a node only continue the BFS from the source with the minimum \id it has seen so far, and ignores the BFS's for all other sources. 
Therefore, the only source that its BFS is not paused after $D$ rounds is $v_{\min}$. 
Each leaf of a BFS tree reports the termination of the construction up the tree, and back to the tree's parent.
Hence, $v_{\min}$ becomes aware that it is the node with the lowest \id in $O(D)$ rounds.
In another $O(D)$ time, $v_{\min}$ disseminates along its BFS tree the depth of this tree, and all nodes learn this value, to which we refer as $D'$. 
Observe that $D' \leq D<2D'$,
and $2D'$ will be used when an upper bound on~$D$ is needed (e.g., for the time of each iteration described below).

We next set $S=\{v_{\min}\}$, and preform $k-1$ iterations. 
In each iteration, a BFS is performed, where the sources are all nodes in $S$.
Each node chooses to join the first BFS tree which reaches to it, breaking ties arbitrarily.
After $O(D)$ rounds, the BFS is done and each node knows its distance to $S$.
We next need to find the node with minimum \id among the furthest nodes to $S$, which we refer to it as $v^*$.
This can be done in $O(D)$ rounds. In each round, every node informs its neighbors which node has the maximum distance among those it is currently aware of, and if it knows more than one node with maximum distance, it only reports the one with the smallest \id. 
After $D'$ such rounds, node $v^*$ knows that it is the furthest node to $S$ with minimum \id. 
Finally, $v^*$ can be added to $S$.
At the end of the algorithm, the nodes in $S$ mark themselves as centers.


The proof for the approximation ratio of our algorithm comes from the approximation ratio of the known greedy approach by \citet{DBLP:journals/tcs/Gonzalez85}, which we applied in our algorithm. In Lemma~\ref{lem:greedy:2approx} we formally state this approximation. 

\begin{lemma}[\cite{DBLP:journals/tcs/Gonzalez85}]
\label{lem:greedy:2approx}
    Let $G$ be a graph, and $k\geq 1$ be an integer. Assume $S_1={v_1}$, where $v_1\in G$ is an arbitrary node. For each $1<i\leq k$, we have $S_i=S_{i-1} \cup \{v_i\}$, where $v_i$ is a node of $G$ with maximum distance to $S_{i-1}$. Then, $S_k$ is a $2$-approximate solution for $k$-center of $G$.
\end{lemma}


We can now summarize this section in Theorem~\ref{thm:congest:alg}.

\begin{theorem}
\label{thm:congest:alg}
    There exists a deterministic $2$-approximation
    algorithm for the $k$-center problem in the \cgst model that needs $O(kD)$ communication rounds.
\end{theorem}

\section{A Lower Bound in the \cgst Model}
\label{sec:cgst lb}
In this section, we show a lower bound on the number of communication rounds for any algorithm for the $k$-center problem in the \cgst model with an approximation ratio better than $4/3$. We start with a lower bound for the $1$-center problem and then extend it for the $k$-center problem for a general $k$.

\subsection{A Lower Bound for the $1$-Center Problem}
\label{sec:lb:congest:1center}
To prove the lower bound, we use a graph that was introduced by \citet*{AbboudCKP21} in order to prove a lower bound for computing the radius of a graph.
Let $x$ and $y$ be two binary strings of length $\len$. 
The graph $G_{x, y}$ is built as follows,
on $n=\Theta(\ell)$ nodes (see \cref{fig:lb:cogest}).
On a high level, the graph consists of two main sets of $\len$ nodes each, $A$ and $B$, and a path on $4$ nodes: $c_A,\bar{c}_A,c_B,\bar{c}_B$.
The input strings $x$ and $y$ for a set-disjointness problem are used to set which edges from $A$ to $\bar{c}_A$ exist (representing $x$), and which edges form $B$ to $\bar{c}_B$ exist (representing $y$).
The rest of the graph is built in order to guarantee that the optimal solution for the $1$-center problem, $\optOne$ satisfies $\optOne=4$ if and only if $x$ and $y$ are disjoint.
Hence, by simulating a $1$-center algorithm and finding the distance from the chosen center to all other nodes, Alice and Bob can decide if $x$ and $y$ are disjoint.
Finally, if the algorithm is too fast, Alice and Bob can simulate it with too little communication, contradicting the communication-complexity lower bound for disjointness.

We now present the graph and the proof in detail.
The graph $G_{x,y}$ consists of node sets $A$, $B$, $F_A$, $T_A$, $F_B$, and $T_B$, as well as nodes $c_A$, $\bar{c}_A$, $c_B$, and $\bar{c}_B$.
The set $A$ is the set of $\len$ nodes $a^0, a^1,\ldots, a^{\len-1}$ and the set $B$ is the set of $\len$ nodes $b^0, b^1,\ldots, b^{\len-1}$.
For $S\in \{A,B\}$, the set $F_S$ consists of $\lfloor \log_2{\len} \rfloor$ nodes $f^{0}_S, f^{1}_S,\ldots, f^{\lfloor\log_2{\len-1} \rfloor}_S$.
Similarly, the set $T_S$ contains node $t^{0}_S, t^{1}_S,\ldots, t^{\lfloor \log_2{\len-1} \rfloor}_S$. 
The edges of the graph $G_{x, y}$ are described in the following.

\begin{itemize}
    \item \textbf{Edges from $A$ to $F_A$ and $T_A$} and \textbf{edges from $B$ to $F_B$ and $T_B$}. Let $0\leq i<\len$ and $0\leq h<\len$, and let $\bin^{i}_h$ be the $h$-th bit in the binary representation of $i$. If $\bin^{i}_h=0$, we connect $a^i$ to $f^h_{A}$ and we connect $b^i$ to $f^h_{B}$. Otherwise, if $\bin^{i}_h=1$, we connect $a^i$ to $t^h_{A}$ and we connect $b^i$ to $t^h_{B}$.
    That is to say, we connect $a^i$ to the binary representation of $i$ in the sets $F_A$ and $T_A$, where $F$ and $T$ represent ``false'' and ``true'';
    $b^i$ is similarly connected.

    \item \textbf{Edges from $\bar{c}_A$ to $A$} and \textbf{edge from $\bar{c}_B$ to $B$}.
    For each $0\leq i < \len$, we connect $a^i$ to $\bar{c}_A$ if and only if $x[i]=1$, and we connect $b^i$ to $\bar{c}_B$ if and only if $y[i]=1$.
    \item \textbf{Other edges from $c_A$, $\bar{c}_A$, $c_B$ and $\bar{c}_B$}.
    We connect $c_A$ to all nodes in $A$ and we connect $c_B$ to all nodes in $B$. We also connect $\bar{c}_A$ to all nodes in $F_A$ and $T_A$, and similarly, we connect $\bar{c}_B$ to all nodes in $F_B$ and $T_B$.
    \item \textbf{Edges between $F_A$, $T_A$, $F_B$, and $T_B$}. For every $0\leq h<\len$, we connect $f^h_A$ to $t^h_A$ and $t^h_B$, and we connect $t^h_A$ to $f^h_B$.
    \item \textbf{Edges from $w^0$, $w^1$ and $w^2$}. We connect $w^0$ to all nodes in $A$, and also $w^0$ to $w^1$ and $w^1$ to~$w^2$.
\end{itemize}

\begin{figure}
    \begin{center}
        \includegraphics[
		trim=3cm 5.5cm 3cm 2.5cm,clip,width=0.95\linewidth]{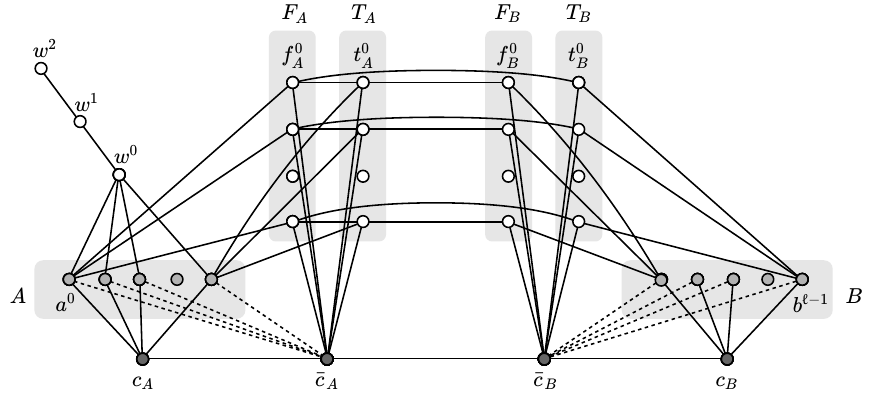}
	\end{center}
    \caption{Illustration of the graph $G_{x, y}$~\cite{AbboudCKP21}.
    The dotted edges depend on the inputs for the disjointness problem.}
    \label{fig:lb:cogest}
\end{figure}

Distances in the graph $G_{x,y}$ presented above have the following properties.

\begin{claim}[{\cite[Claim 4]{AbboudCKP21}}]
\label{claim:Gxy}
    Let $G_{x,y}$ be the graph defined above, and $V$ be the set of its nodes. Then, the following holds
    \begin{enumerate}
        \item \label{claim:Gxy:1} For every node $u\in V\setminus A$, we have $\ecc(u) \geq 4$.
        \item For every node $a^i\in A$ and every $u\in V\setminus \{b^i, c_B\}$ we have $\dist(a^i,u) \leq 3$. \label{claim:Gxy:2}
    \end{enumerate}
\end{claim}

The first part of the claim holds since the distance of $w^2$ to any node in $V\setminus A \cup \{w^0, w^1, w^2\}$ is at least $4$, and the distance of $c_B$ to any node in $\{w^0, w^1, w^2\}$ is at least $4$. 
To show the second part of the claim, observe that the distance between $a^i$ and any node in $V\setminus B\cup \{c_B\}$ is clearly at most $3$. Besides, if $j\neq i$, there exist $0\leq h<\len$ such that $\bin^i_h \neq \bin^j_h$, where $\bin^i_h$ and $\bin^j_h$ are the $h$-th bit in the binary representation of $i$ and $j$, respectively. If $\bin^i_h=0$ and $\bin^j_h=1$, the path $(a^i,f^h_A,t^h_B,b^j)$ exists. Otherwise, if $\bin^i_h=1$ and $\bin^j_h=0$, the path $(a^i,t^h_A,t^h_B,b^j)$ exists. It means that $\dist(a^i, b^j) \leq 3$, and therefore, the second part of the claim holds.
The rigorous proof of this claim can be found in~\cite[Claim 4]{AbboudCKP21}.

Recall that  $x$ and $y$ are not disjoint if there is a $0\leq i < \len$ such $x[i]=y[i]=1$,
and otherwise, they are disjoint.
In Lemma~\ref{lem:disjoint:4:3}, we show that $\optOne(G_{x,y})=4$ if and only if $x$ and $y$ are disjoint.

\begin{lemma}
\label{lem:disjoint:4:3}
    If the strings $x$ and $y$ are disjoint then $\optOne(G_{x,y})=4$.
    Otherwise, $\optOne(G_{x,y})=3$.
\end{lemma}

\begin{proof}
    To prove the lemma, we first show that if $\optOne < 4$, then an optimal center of $G_{x,y}$ for the $1$-center problem is a node $a^i\in A$, such that $x[i]=y[i]=1$.
    Assume that $\optOne < 4$. Then Claim~\ref{claim:Gxy}(\ref{claim:Gxy:1}) implies that the optimal center is in $A$. Let $a^i\in A$ be an optimal center. Then, $\dist(a^i, b^i) \leq 3$ should holds since we assume $\optOne < 4$.
    But $\dist(a^i, b^i) \leq 3$ holds only if $a^i$ is connected to $\bar{c}_A$ and $b^i$ is connected to $\bar{c}_B$, which means $x[i]=y[i]=1$.

    We next consider the case that $x$ and $y$ are disjoint. 
    In this case, we have $\optOne \geq 4$ since otherwise, our claim implies that there exists $0\leq i<\ell$ such that $x[i]=y[i]=1$ and contradicts the disjointness of $x$ and $y$. Hence, $\optOne \geq 4$ holds. On the other hand, all nodes are within distance~$4$ of $a^0$, and therefore, $\optOne = 4$.
    
    Now, we consider the case that $x$ and $y$ are not disjoint. This means that there exist $0\leq i < \ell$
    such that $x[i]=y[i]=1$. We show that $\ecc(a^i)=3$ and then $\optOne=3$. Claim~\ref{claim:Gxy}(\ref{claim:Gxy:2}) 
    states that $\dist(a^i,u)\leq 3$ holds for any $u\in V\setminus \{b^i, c_B\}$. Besides, $a^i$ is connected to $\bar{c}_A$ as $x[i]=1$ and $b^i$ is connected to $\bar{c}_B$ since $y[i]=1$. Thus, $\dist(a^i, b^i)=\dist(a^i, c_B)=3$. Putting everything together we have $\optOne = 3$.
\end{proof}

\begin{theorem}
	\label{thm:lb-for-1-cntr}
Any \cgst algorithm that returns an $\alpha$-approximation for the $1$-center problem with $\alpha<4/3$, even randomized,
must take $\Omega(n/\log^2 n)$ rounds to complete.
\end{theorem}

\begin{proof}
	Assume for contradiction a \cgst algorithm faster than in the theorem's statement.
	Let $x,y\in \{0,1\}^\ell$ be two inputs of Alice and Bob for the disjointness problem.
	Alice and Bob simulate the algorithm on the graph $G_{x,y}$ described above.
	
	To this end, we split the graph node by $V_A=A\cup F_A\cup T_A \cup\{c_A,\bar C_a, w^0,w^1,w^2\}$ and $V_B=V\setminus V_A$
	(the nodes on the left and right sides of \cref{fig:lb:cogest}, respectively).
	Alice is in charge of the nodes of $V_A$ and Bob on $V_B$.
	To simulate a round, they locally simulate the exchange of messages in each of their node sets, and exchange bits to simulate the messages between nodes of $V_A$ and $V_B$.
	Since there are $O(\log n)$ edges between the sets, a simulation of a round requires $O(\log^2 n)$ bits of communication.
	
	If the output of the algorithm is a node not in $A$, Alice and Bob return $1$ for the disjointness problem.
	If it is some node $a^i\in A$,
	the exchange the bits $x[i]$ and $y[i]$ (for the same index $i$) and return $0$ for the disjointness problem if and only if $x[i]=y[i]=1$.
	
	For correctness, first note that Alice and Bob return $0$ for the disjointness problem only if they find an index such that $x[i]=y[i]=1$, so this answer must be correct.
	When they return $1$, on the other hand, it might be since the algorithm returned a center $u\notin A$, or a center $a^i$ such that $x[i]=0$ or $y[i]=0$ (or both).
	
	If Alice and Bob return $1$ because of a center $u\notin A$, then \cref{claim:Gxy}(\ref{claim:Gxy:1}) guarantees that $\ecc(u)\geq 4$,
	and by $\alpha<4/3$ we get that $\optOne>3$.
	Have the sets not been disjoint, \cref{lem:disjoint:4:3} guarantees that $\optOne=3$, a contradiction.
	
	If Alice and Bob return $1$ because of a center $a^i\in A$, then we also know that $x[i]=0$ or $y[i]=0$.
	Note that $\dist(a^i,b^i)\geq 4$:
	the nodes of $\bin(a^i)$ and $\bin(b^i)$ are never neighbors, so any $3$-path connecting $a^i$ and $b^i$ must go through the edge $(\bar c_A,\bar c_B)$; but since $x[i]=0$ or $y[i]=0$, no such path can exists.
	Hence, $\ecc(a^i)\geq 4$.
	As before, 
	if the sets were not disjoint, \cref{lem:disjoint:4:3} guarantees that $\optOne=3$, and we would have got a contradiction.
	
	For the complexity,
	let $T$ be the number of rounds used by the algorithm.
	To simulate these rounds, Alice and Bob exchange $O(T\log^2n)$ bits.
	As shown above, they solve the disjointness problem on $\ell$ bits, which requires them to communicate $\Omega(\ell)$ bits, even when using a randomized algorithm.
	As $\ell=\Theta(n)$, we get
	$T=\Omega(n/\log^2 n)$.
	In fact, Alice and Bob must also exchange the id of the node chosen as the center, but these $\log n$ bits of communication do not affect the asymptotic complexity.
	
	Finally, note that Alice and Bob cannot solve disjointness solely by the output of the algorithm, and need extra communication after it.
    Hence, we cannot utilize standard reductions such as~\cite[Theorem~6]{AbboudCKP21} and use a non-standard one.
\end{proof}

\subsection{Extending the Lower Bound to the $k$-Center Problem}
\label{sec:lb:congest:kcenter}

To generalize the lower bound for the $1$-center problem to the $k$-center problem, we consider $k$ copies of the graph $G_{x,y}$ introduced in Section~\ref{sec:lb:congest:1center},
such that all copies share the node $w^2$. We refer to this graph as $G^k_{x,y}$.

\begin{claim}
\label{cliam:lb:1in:each:copy}
    In any $\alpha$-approximate solution for $k$-center of $G^k_{x,y}$, $\alpha<3/2$, exactly one center is chosen from each copy.
\end{claim}

\begin{proof}
    Note that
    \cref{lem:disjoint:4:3} guarantees that in each copy we can choose a center and get a solution with distance at most $4$ from each node to a center, and $\opt\leq4$.
    If some copy contains no center, then $b^0$ in this copy has distance at least $6$ to $w^2$, and the solution is at least $(6/4)$-approximation, a contradiction.
\end{proof}

\begin{theorem}
	Any \cgst algorithm that returns an $\alpha$-approximation for the $k$-center problem with $\alpha<4/3$, even randomized,
	must take $\Omega(n/(k^2 \log^2 n))$ rounds to complete.
\end{theorem}
The proof extends the proof of \cref{thm:lb-for-1-cntr} by using the graph $G^k_{x,y}$ described above.
Note that all copies are for the same input pair $(x,y)$.

\begin{proof}
	Assume a faster algorithm than in the theorem statement.
	Alice and Bob simulate this algorithm for $\ell$ input bits, where $\ell=\Theta(n/k)$,
	and get an output at the form of $k$ center nodes.
	If none of the output nodes is of the form $a^i\in A$ for some copy, Alice and Bob answer $1$ for the disjointness problems.
	Otherwise, for every node $a^i\in A$ in the output (perhaps from different copies), they exchange $x[i],y[i]$,
	and answer $0$ for the disjointness problem only if $x[i]=y[i]=1$ for at least one such $i$.
	
	For correctness, recall first that \cref{cliam:lb:1in:each:copy} guarantees that each copy contains exactly one center.
	
	If Alice and Bob answer $0$, this is because they found an index $i$ of intersection and the algorithm is correct.
	If they answer $1$, the proof follows the same lines of the proof of \cref{thm:lb-for-1-cntr}: 
	if the sets are not disjoint then \cref{lem:disjoint:4:3} guarantees each copy has a center $a^i$ with $\ecc(a^i)\leq 3$, and we get $\opt\leq 3$.
	On the other hand, the algorithm returns at each copy a node $u\notin A$ or a node $a^i$ with $x[i]=0$ or $y[i]=0$, which has $\dist(a^i,b^i)\geq 4$;
	in both cases, the solution is not an $\alpha$-approximation for $\alpha<4/3$.
	
	For the complexity, assume the algorithm takes $T$ rounds.
	Alice and Bob communicate, for 
	each copy and each round, $O(\log (n/k) \log n)$ bits, which are $O(\log^2n)$ bits.
	At the end of the algorithm, they may exchange the indices of at most $k$ centers of the form $a^i$ and their inputs $x[i]$ and $y[i]$ in these locations, with takes requires another $O(k\log n)$ bits.
	In total, they exchange
	$O(Tk\log^2 n +k\log n)$ bits,
	which is $O(Tk\log^2 n)$.

	On the other hand, they solve disjointness on $\ell=\Theta(n/k)$ input bits,
	so they must communicate $\Omega (n/k)$ bits.
	Hence,
	\[
		Tk\log^2 n\geq cn/k
	\]
	for some constant $c$, and we get
	\[
	T=\Omega\left( \frac{n}{k^2\log^2n}\right)
	\]
	as claimed.
\end{proof}

\section{A $2$-Approximation Algorithm in the \clq Model}
\label{sec:clique}
We now present Algorithm~\ref{alg:clique}, which provides an approximate solution to the $k$-center clustering problem in the \clq model.
It is applicable both when the distance between the nodes are weighted and when they are unweighted.

The algorithm consists of two phases. 
The first phase is computing all-pairs shortest distances or approximating them, using known algorithms.

The second phase is to greedily find the centers, which is done in additional $k$ rounds.
In this phase, we first find $v_{\min}$, the node with the minimum \id.
This is trivial when the ids are in $\{1,\ldots,n\}$, but easy also without this assumption, by having each node send its \id to all other nodes.
We then set $S\gets\{v_{\min}\}$. Next, we have $k-1$ iteration,
and in each iteration we find $v^*$, which is the furthest node from $S$ (if there is more than one furthest node, we define $v^*$ as the furthest node with minimum \id). 
To accomplish this, it is enough that each node sends its distance to $S$ to all nodes. 
Hence, all nodes can know who $v^*$ is in one communication round. 
At the end of each iteration, we set $S$ to $S\cup \{v^*\}$, and report $S$ as the set of centers at the end of the algorithm.
This process is formalized in the following algorithm.

\begin{algorithm}[H]
\caption{Varying approximation ratios in the \clq model}\label{alg:clique}
\begin{algorithmic}[1]
    \State Compute all-pairs shortest paths
    \State Let $v_{\min}$ be the node with minimum \id
    \State $S\gets \{v_{\min} \}$
    \For{$k-1$ times}
        \State Let $v^*$ be the furthest node to $S$
        \State $S \gets S \cup \{v^*\}$
    \EndFor
    \State Each node in $S$ marks itself as a center
\end{algorithmic}
\end{algorithm}
 

If we calculate the exact all-pairs shortest paths in the first phase of Algorithm~\ref{alg:clique}, Lemma~\ref{lem:greedy:2approx} indicates that the set $S$ computed by the algorithm is a $2$-approximate solution for \kc. 

Approximate distance can be computed much faster than exact ones, and we next show that 
Algorithm~\ref{alg:clique} can also work with approximate distances, in which case it computes an approximate \kc solution.
To prove this, we now prove  Lemma~\ref{lem:greedy:2+eps}, which extends Lemma~\ref{lem:greedy:2approx} to the scenario where only a multiplicative $\alpha$-approximation with one-sided error of all-pairs shortest paths is computed in the first phase.

\begin{lemma}
\label{lem:greedy:2+eps}
    For any $\alpha\geq 1$, if an $\alpha$-approximation of all-pairs shortest paths is computed in the first phase of Algorithm~\ref{alg:clique}, then the set $S$ obtained by the algorithm is a $(2\alpha)$-approximate solution for the $k$-center problem.
\end{lemma}

\begin{proof}
    If $n\leq k$, the algorithm trivially returns all the nodes as the centers, so we assume $n > k$ in the following. 
    Let $S$ be the set returned by Algorithm~\ref{alg:clique}, $\hat{v} \in V\setminus S$ the furthest node from~$S$, and $\hat{r}$ its distance. 

    First, we claim that $\opt \geq \hat{r}/(2\alpha)$. 
    Since $\hat{v}$ is not added to $S$, we can conclude that the distance between any two nodes in $S$ is at least $\hat{r}/\alpha$. Thus, the distance between any two nodes in $S\cup \{\hat{v}\}$ is at least $\hat{r}/\alpha$ as well. 
    In addition, in any optimal solution there exist two nodes in $S\cup \{\hat{v}\}$ that have the same nearest center in that optimal solution since $|S \cup \{\hat{v}\}|=k+1$. 
    Since the distance between these two nodes is at least $\hat{r}/\alpha$, at least one of them has distance at least $\hat{r}/(2\alpha)$ from their common center, implying $\opt \geq \hat{r}/(2\alpha)$.
    On the other hand, all nodes in $V$ are within distance $\hat{r}$ from $S$ by the definition of $\hat{r}$. Hence, $S$ is a solution with the approximation ratio of at most $\frac{\hat{r}}{\hat{r}/(2\alpha)}=2\alpha$. 
\end{proof}

All is left now is to plug fast \clq distance computation algorithms~\cite{ChechikZ22,Censor-HillelDK21,DoryP22,Censor-HillelKK19} in the lemma, and we get fast algorithms for exact and approximate \kc.
This yields the results detailed in 
Theorem~\ref{thm:clique:alg} and \cref{tab:kc times}, showing that \kc can be approximated in the same times as all pairs shortest paths computation, up to an additive $O(k)$ time.
Note that for the specific case of $(2+\epsilon)$-approximation, a much faster (randomized) algorithm exists~\cite{BandyapadhyayIP22}, running in $O(\poly\log n)$ rounds.

\begin{theorem}
\label{thm:clique:alg}
    There exists a deterministic \clq algorithm
    for the $k$-center problem which gives a 
    $2$-approximation
    in $O(n^{1/3}+k)$ rounds on weighted graphs
    and
    $O(n^{0.158}+k)$ rounds on unweighted graphs.
    
    For every constant $0<\epsilon<1$, there exist deterministic algorithms
     that give a 
    $(2+\epsilon)$-approximation
    in $O(n^{0.158}+k)$ rounds on unweighted graphs,
    a $(4+\epsilon)$-approximation
    in $O(\log^4 \log n +k)$ rounds on unweighted graphs, and
    a $(6+\epsilon)$-approximation in
     $O(\log^2 n +k)$ rounds on weighted graphs.
    
    There exists a randomized algorithm
    which gives a 
    $O(\log n)$-approximation
    in $O(k)$ rounds on both weighted and unweighted graphs.
\end{theorem}

\begin{table}[tb]
    \caption{Approximation times in the \clq model, for constant $0<\epsilon<1$}
    \label{tab:kc times}
    \centering
    \begin{tabular}{ccccc} \hline
         Approx. ratio  &Time&Weighted?&Deter.?&  Ref.\\ 
		\hline
         2  &$O(n^{1/3}+k)$&No&Yes&  \cite{Censor-HillelKK19} \& Thm.~\ref{thm:clique:alg}\\ 
          2 &$O(n^{0.158}+k)$&Yes&Yes&  \cite{Censor-HillelKK19} \& Thm.~\ref{thm:clique:alg}\\
          $2+\epsilon$ &$O(n^{0.158}+k)$&No&Yes&  \cite{Censor-HillelKK19} \& Thm.~\ref{thm:clique:alg}\\
 $2+\epsilon$& $O(\poly\log n)$& Yes& No&\cite{BandyapadhyayIP22}\\ 
          $4+\epsilon$ &$O(\log^4\log n+k)$& No&Yes&  \cite{DoryP22} \& Thm.~\ref{thm:clique:alg}\\ 
          $6+\epsilon$&$O(\log^2 n+k)$& Yes&Yes&  \cite{Censor-HillelDK21} \& Thm.~\ref{thm:clique:alg}\\ 
 $O(\log n)$ &$O(k)$& Yes& No& \cite{ChechikZ22} \& Thm.~\ref{thm:clique:alg}\\ \hline
    \end{tabular}
\end{table}

\subsubsection*{Acknowledgments} 
We thanks Morteza Monemizadeh for many valuable conversations on this project, Michal Dory for discussions regarding distance computation in the \clq model, and the anonymous reviewers of SIROCCO'24 for their helpful comments.

\bibliographystyle{abbrvnat}
\bibliography{kcenters-bib}

\end{document}